\documentclass[12pt,leqno,letterpaper]{article}

\usepackage{amsmath,amsthm,enumerate,amssymb,mathtools}
\usepackage[utf8]{inputenc}

\usepackage[english]{babel}

\newtheorem{theorem}{Theorem}[section]
\newtheorem{proposition}[theorem]{Proposition}
\newtheorem{lemma}[theorem]{Lemma}
\newtheorem{corollary}[theorem]{Corollary}

\newcommand{\tr}{{\rm Tr\hskip -0.2em}~}
   \newcommand{\df}[2]{\frac{d#1}{d#2}}

\begin{document}

\title{Geometric properties for a class of deformed trace functions}
\author{Frank Hansen}
\date{}

\maketitle

\begin{abstract} We investigate geometric properties for a class of trace functions expressed in terms of the deformed logarithmic and exponential functions. We extend earlier results of Epstein, Hiai, Carlen and Lieb.
\\[1ex]them 
{\bf MSC2010} classification: 47A63; 15A90\\[1ex]
\end{abstract}

\section{Preliminaries}

Tsallis~\cite{tsallis:2009} generalised in 1988 the standard Bolzmann-Gibbs entropy to a non-extensive quantity $ S_p(\rho) $ depending on a parameter $ p. $ In the quantum version it is given by
\[
S_p(\rho)=\frac{1-\tr\rho^p}{p-1}\qquad p\ne 1,
\]
where $ \rho $ is a density matrix. It has the property that
$
S_p(\rho)\to S(\rho)
$
for $ p\to 1, $ where $ S(\rho)=-\tr\rho\log\rho $ is the von Neumann entropy. 

\subsection{The deformed logarithm and exponential}
The Tsallis entropy may be written on a similar form
\[
S_p(\rho)=-\tr\rho\log_p \rho,
\]
where the deformed logarithm $ \log_p $ defined for positive $ x $ is given by
\[
\log_p x=\int_1^x  t^{p-2}\,dt = \left\{\begin{array}{ll}
                                                         \displaystyle\frac{x^{p-1}-1}{p-1}\quad &p\ne 1\\[2.5ex]
                                                         \log x                                 &p=1.
                                                         \end{array}\right.
\]
The deformed logarithm is also denoted the $ p $-logarithm. The range of the $ p $-logarithm is given by the intervals
\[
\begin{array}{ll}
 \bigl(-(p-1)^{-1}, \infty\bigr)\quad&\text{for}\quad p>1\\[1.5ex]
  \bigl(-\infty, -(p-1)^{-1}\bigr)&\text{for}\quad p<1\\[1.5ex]
  \bigl(-\infty,\infty\bigr)&\text{for}\quad p=1.
 \end{array}
 \]
 The inverse function $ \exp_p $ (denoted the $ p $-exponential) is always positive and given by 
\[
\exp_p(x)=
\left\{\begin{array}{ll}
(x(p-1)+1)^{1/(p-1)}\quad&\text{for $ p>1 $ and $x>-(p-1)^{-1} $}\\[1.5ex]
(x(p-1)+1)^{1/(p-1)}\quad&\text{for $ p<1 $ and $x<-(p-1)^{-1} $}\\[1.5ex]
\exp x \quad&\text{for $ p=1$ and  $x\in\mathbf R. $ }
\end{array}\right.
\]
The $ p $-logarithm and the $ p $-exponential functions converge, respectively, to the logarithmic and the exponential functions for $ p\to 1. $
We note that
\begin{equation}\label{derivative of $ q-exponential}
\df{}{x}\log_p(x)=x^{p-2}\qquad\text{and}\qquad \df{}{x}\exp_p(x)=\exp_p(x)^{2-p}\,.
\end{equation}
We will also need the following lemma.

\begin{lemma}\label{lemma: exponent of deformed exponential}
Take arbitrary $ p\in\mathbf R. $ Independent of $ x> 0, $ we have
\[
\log_p x^q=q \log_\alpha x,
\]
where $ \alpha=1+q(p-1). $ Furthermore, take arbitrary $ q\ne 0 $ and set   $ \beta=1+(p-1)/q\,. $ For any $ x\in\mathbf R $ in the domain of $ \exp_q\,, $ we obtain that $ qx $ is in the domain of $ \exp_\beta $ and that
\[
(\exp_p x)^q=\exp_\beta(q x).
\]
\end{lemma}

\begin{proof}
We substitute $ u=t^{1/q} $ (thus $ t=u^q) $ in
\[
\log_p x^q=\int_1^{x^q} t^{p-2}\,dt
\]
and note that $ du=q^{-1} t^{(1-q)/q}\, dt. $ Therefore, $ dt=q t^{(q-1)/q}\, du $ and thus
\[
\log_p x^q=\int_1^x u^{q(p-2)} q u^{q-1}\,du=q\int_1^x u^{q(p-1)-1}\, du.
\]
Since $ q(p-1)-1=q(p-1)+1-2=\alpha-2, $ the first statement follows. The definition of $ \beta $ implies $ q/(p-1)=1/(\beta-1). $ There are now four cases depending on $ q>0, $ $ q<0, $ $ p>1 $ and $ p<1. $ In all four cases it follows that $ qx $ is in the domain of $ \exp_\beta. $ We finally obtain from the first result in the lemma that
\[
\log_\beta(\exp_p x)^q=q\log_p(\exp_p x) =qx
\]
and therefore $ (\exp_p x)^q = \exp_\beta(qx). $
\end{proof}

\subsection{Convexity and min-max theorems}

An important tool in our investigation is taken from convex analysis. These techniques are used in engineering, automatic control, signal processing, resource allocation, portfolio theory, and numerous other fields. We are in particular using that partial minimisation of a convex function is convex \cite[Section 3.2.5]{Boyd2004}.
This technique was successfully applied by Carlen and Lieb in the investigation of trace functions \cite{carlen:2008}. We provide the proof as a convenience to the reader.

\begin{lemma} Let $ f\colon X\times Y\to \mathbf R $ be a function of two variable and set 
\[
g(y)=\inf_{x\in X} f(x,y)\quad\text{and}\quad h(y)=\sup_{x\in X} f(x,y)
\]
for $ x\in X. $

\begin{enumerate}[(i)]

\item If $ f(x,y) $ is jointly convex, then $ g $ is convex.

\item If $ f(x,y) $ is convex in the second variable, then $ h $ is convex.

\item If $ f(x,y) $ is concave in the second variable, then $ g $ is concave.

\item If $ f(x,y) $ is jointly concave, then $ h $ is concave.

\end{enumerate}

\end{lemma}

\begin{proof}
Take $ \varepsilon>0 $ and elements $ y_1,y_2\in Y. $ Pick $ x_1,x_2\in X $ such that
\[
g(y_1)\ge f(x_1,y_1)+\varepsilon\quad\text{and}\quad g(y_2)\ge f(x_2,y_2)+\varepsilon.
\]
Then
\[
\begin{array}{l}
g(\lambda y_1+(1-\lambda) y_2)\le f(\lambda x_1+(1-\lambda) x_2, \lambda y_1+(1-\lambda) y_2)\\[1.5ex]
\le\lambda f(x_1,y_1)+(1-\lambda) f(x_2,y_2)\le \lambda g(y_1)+(1-\lambda) g(y_2) -\varepsilon,
\end{array}
\]
so $ g $ is convex. Pick to $ \lambda y_1+(1-\lambda)y_2\in Y $ and $ \varepsilon >0 $ an element $ z\in X $ such that
\[
h(\lambda y_1+(1-\lambda) y_2) \le f(x,\lambda y_1+(1-\lambda) y_2)+\varepsilon.
\]
Then
\[
\begin{array}{l}
h(\lambda y_1+(1-\lambda) y_2)\le \lambda f(x,y_1)+(1-\lambda) f(x,y_2)+\varepsilon\\[1.5ex]
\le \lambda h(y_1)+(1-\lambda)h(y_2)+\varepsilon,
\end{array}
\]
so $ h $ is convex. The cases $ (iii) $ and $ (iv) $ follow by considering $ -f(x,y). $
\end{proof}

\subsection{The Young tracial inequalities} 

The following inequalities are known as the tracial Young inequalities. We prefer to prove them as below.

\begin{proposition}\label{basic tracial inequalities}
Let $ A $ and $ B $ be positive definite matrices. Then
\[
\tr A^p B^{1-p} \le p\tr A+(1-p)\tr B\qquad 0\le p\le 1
\]
and
\[
\tr A^p B^{1-p} \ge p\tr A+(1-p)\tr B\qquad p\le 0, \,  p\ge 1.
\]
\end{proposition}

\begin{proof}
Let first $ 0\le p\le 1. $ We may write
\[
\begin{array}{rl}
\tr A^pB^{1-p}&=\tr L_A^p R_B^{1-p} I=S_f^I(A,B)\\[1.5ex]
&\le\tr \bigl(pL_A+(1-p)R_B\bigr)I=p\tr A+(1-p)\tr B,
\end{array}
\]
where $ f(t)=t^p $ for $ t>0, $ and $ L_A $ and $ R_B $ are the left and right multiplication operators. The first equality above in terms of the quasi-entropy $ S_f^I(A,B) $ follows since $ L_A $ and $ R_B $ commute, and the first inequality in the proposition then follows from the geometric-arithmetic mean inequality. Since Jensen's inequality reverses for the extensions of a chord (corresponding to the cases $ p\le 0 $ or $ p\ge 1), $ the second inequality of the proposition follows.
\end{proof}

\section{Variational expressions}  

We take the following variational representations from our paper \cite[Lemma 2.1]{shi-hansen:2020} with a slightly simplified proof.

\begin{proposition}\label{first variational expression in terms of the deformed logarithm} 
For positive definite operators $X$ and $ Y$ we have
\[
\tr Y=\left\{
\begin{array}{ll}
\displaystyle\max_{X>0}\bigl\{\tr X-\tr X^{2-q}\left(\log_q X-\log_q Y\right)\bigr\}, \quad &q\le 2,\\[2.5ex]
\displaystyle\min_{X>0}\bigl\{\tr X-\tr X^{2-q}\left(\log_q X-\log_q Y\right)\bigr\},  &q>2.
\end{array}
\right.
\]
\end{proposition}

\begin{proof}  We learned in Proposition~\ref{basic tracial inequalities} that
\[
\begin{array}{rl}
\tr X^{p}Y^{1-p} &\le  p \tr X+(1-p)\tr Y,\qquad 0\le p\le 1,\\[2ex]
\tr X^{p}Y^{1-p} &\ge  p \tr X+(1-p)\tr Y,\qquad p\le 0,\, p\ge 1.
\end{array}
\]
By combining the first inequality for $ 0\le p<1 $ with the case $ p>1 $ in the second, we obtain
\[
\tr Y \ge \tr X-\frac{\tr X-\tr X^pY^{1-p}}{1-p}\,, \qquad p\ge 0,\, p\ne 1,
\]
while the case $ p\le 0 $ gives the inequality
\[
\tr Y \le \tr X-\frac{\tr X-\tr X^pY^{1-p}}{1-p}\,, \qquad p\le 0.
\]
For $X=Y$ the above inequalities become equalities.
Setting $ q=2-p, $ the first range ($ p\ge 0, $ $ p\ne 1) $ is transformed to the range $( q\le 2, $ $ q\ne 1), $ while the second range $ (p\le 0) $ is transformed to the range $ (q\ge 2). $ Since $ p=2-q $ and $ 1-p=q-1 $ we obtain
\[
\tr Y=\left\{
\begin{array}{ll}
\displaystyle\max_{X>0}\Bigl\{\tr X-\frac{\tr X^{2-q}\left(X^{q-1}-Y^{q-1}\right)}{q-1}\Bigr\}, \quad & q\in(-\infty,2],\, q\ne 1, \\[2.5ex]
\displaystyle\min_{X>0}\Bigl\{\tr X-\frac{\tr X^{2-q}\left(X^{q-1}-Y^{q-1}\right)}{q-1}\Bigr\}, &q\in [2,\infty).
\end{array}
\right.
\]
By using the definition of the deformed logarithm we note that
\[
\frac{X^{q-1}-Y^{q-1}}{q-1}=\log_q(X)-\log_q(Y),
\]
and by inserting this in the expressions above, we obtain the desired statements of the proposition, except for $ q=1. $ 
We may finally let $ q $ tend to one in the first inequality and obtain the variational expression
\[
\tr Y=\max_{X>0}\Bigl\{\tr X-{\tr X\left(\log X-\log Y\right)}\Bigr\}
\]
by continuity. This completes the proof.
\end{proof}

Note that the last statement in the above proof entails the inequality 
\[
S(X\mid Y)\ge \tr (X-Y)
\]
for the relative quantum entropy $ S(X\mid Y). $

\subsection{Further preliminaries}

\begin{lemma}\label{the psi-function}
Let $ H $ be an arbitrary matrix, take $ L\ge 0, $ and choose exponents $ p $ and $ s $ such that $ s >0. $
We consider the trace function
\[
\psi_{L,H}(A)=\tr\bigl(L+H^*A^pH\bigr)^s
\]
defined in positive definite matrices.  $ \psi_{L,H}(A) $ is convex (respectively concave) for arbitrary $ H $ and $ L\ge 0, $ if and only if it is convex (respectively concave) for arbitrary $ H $ and $ L=0. $ 
\end{lemma}

\begin{proof}  By considering block matrices
\[
\hat H=\begin{pmatrix}
               L^{1/2} & 0\\
               H & 0
               \end{pmatrix}\quad\text{and}\quad
               \hat A=\begin{pmatrix}
               I & 0\\
               0 & A
               \end{pmatrix}
\]   
we obtain 
\[
\hat H^*\hat A^p \hat H=\begin{pmatrix}
               L+H^*A^pH & 0\\
               0 & 0
               \end{pmatrix}.
\]   
Since $ s>0 $ we obtain in addition
\[
\bigl(\hat H^*\hat A^p \hat H\bigr)^s=\begin{pmatrix}
               \bigl(L+H^*A^pH\bigr)^s & 0\\
               0 & 0
               \end{pmatrix},
\]   
since it is meaningful to set $ 0^s=0. $ Indeed, for $ \varepsilon>0 $ we have
\[
\varepsilon^s=\exp\bigl(s\log\varepsilon\bigr),
\]
and this quantity tends to zero as $ \varepsilon\to 0. $ Therefore,
$ \psi_{L,H}(A)=\psi_{0,\hat H}(\hat A) $ and the statement follows. 
\end{proof}

If $ s<0 $ there exist examples in two by two matrices such that $ \psi_{0,H}(A) $ is convex while $ \psi_{L,H}(A) $ is not.

\section{The main trace function} 

Let $ H $ be an invertible contraction and $ A $ positive definite. Then
\[
H^*\log_p(A)H > \frac{-1}{p-1}H^*H\ge \frac{-1}{p-1}\quad\text{for}\quad p>1
\]
and
\[
H^*\log_p(A)H < \frac{-1}{p-1}H^*H\le \frac{-1}{p-1}\quad\text{for}\quad p<1.
\]
Therefore, $ H^*\log_p(A)H $ belongs to the domain of the $ p $-exponential. This is true even if $ H $ is not invertible since $ \exp_p(0)=1. $ 
Therefore,
\[
\exp_p\bigl(L+H^*\log_p(A)H\bigr)
\]
is well-defined and positive for arbitrary contractions $ H $ and $ p\ne 1, $ provided $ L\ge 0 $ when $ p>1, $ and $ L\le 0 $ when $ p<1. $ In both cases we may define the trace function
\begin{equation}\label{Main deformed trace function}
\varphi^L_{p,q}(A)=\tr\left[\exp_p\bigl(L+H^*\log_p(A)H\bigr)^q\right]
\end{equation}
for arbitrary exponents $ q. $ We furthermore obtain the expression 
\begin{equation}\label{Main deformed trace function, second expression}
\begin{array}{rl}
\varphi_{p,q}^L(A)&=\displaystyle\tr\bigl[I+(p-1)L+(p-1)H^*\frac{A^{p-1}-I}{p-1}H\bigr]^{q/(p-1)}\\[2.5ex]
&=\displaystyle\tr\bigl[I-H^*H+(p-1)L+H^*A^{p-1}H\bigr]^{q/(p-1)}.
\end{array}
\end{equation}
Note that $ (p-1)L\ge 0 $ in both cases. By using Lemma~\ref{the psi-function} we obtain the following:

\begin{corollary}\label{the connection with Upsilon}
Suppose $ q/(p-1)>0. $ Then $ \varphi_{p,q}^L(A) $ is convex (concave) if and only if the trace function $ A\to\tr \bigl(H^*A^p H\bigr)^{q/(p-1)}
 $ is convex (concave).
\end{corollary}

We shall finally explore yet another expression for the main trace function. Given the expression in (\ref{Main deformed trace function}) and setting 
\begin{equation}\label{definition of beta}
\beta=1+\frac{p-1}{q}
\end{equation}
we obtain
\begin{equation}\label{Main deformed trace function, third expression}
\varphi_{p,q}^L(A)=\tr\exp_\beta\bigl(qL+qH^*\log_p(A)H\big),
\end{equation}
where we used Lemma~\ref{lemma: exponent of deformed exponential}.
By replacing $ q $ with $ \beta $ in Proposition~\ref{first variational expression in terms of the deformed logarithm} and setting
\begin{equation}\label{definition of F(X,A)}
F(X,A)=\tr X-\tr X^{2-\beta}\left(\log_\beta X-\log_\beta Y\right), 
\end{equation}
where $ Y=\exp_\beta\bigl(qL+qH^*\log_p(A)H\big), $ we obtain that
\begin{equation}
\varphi^L_{p,q}(A)=\left\{
\begin{array}{ll}
\displaystyle\sup_{X>0} F(X,A) \quad &\beta\le 2,\\[2.5ex]
\displaystyle\inf_{X>0} F(X,A)  &\beta>2.
\end{array}
\right.
\end{equation}
This is the main variational expression to be used. We next calculate
\[
\begin{array}{l}
F(X,A)=\tr X-\tr X^{2-\beta}\left(\log_\beta X-\log_\beta Y\right)\\[2ex]
=\tr X-\tr X^{2-\beta}\left(\log_\beta X-qL-q H^*\log_p(A) H\right)\\[2ex]
=\displaystyle \tr X-\tr X^{2-\beta}\left(\frac{X^{\beta-1}-I}{\beta-1}-qL-qH^*\frac{A^{p-1}-I}{p-1}H\right)\\[3ex]
=\displaystyle \tr X-\frac{1}{\beta-1}\tr\bigl(X-X^{2-\beta}-(p-1)L-X^{2-\beta}H^*(A^{p-1}-I)H\bigr)\\[3ex]
=\displaystyle\left(1-\frac{1}{\beta-1}\right)\tr X+G(X,A),
\end{array},
\]
where we used $ q/(p-1)=1/(\beta-1) $ and set
\begin{equation}
G(X,A)=\frac{1}{\beta-1}\tr \Bigl(X^{2-\beta}(I-H^*H+(p-1)L)+ X^{2-\beta}H^* A^{p-1}H\Bigr).
\end{equation}
The first term in $ F(X,A) $ is linear, so we only have to consider convexity or concavity of $ G(X,A). $ Note as before that $ (p-1)L\ge 0. $

\begin{lemma}\label{lemma: the symmetry}
Let $ H $ be a contraction and consider for arbitrary real $ q $ the trace function $ \varphi^L_{p,q}(A) $ defined in (\ref{Main deformed trace function}).
If $ q/(1-p)>0 $ and $ \varphi_{p,q}^L(A) $ is convex (respectively concave) for arbitrary contractions $ H, $ then so is the trace function $ \varphi_{2-p,q}^{-L}(A). $ 
\end{lemma}

\begin{proof} We may without loss of generality assume that $ H $ is invertible.
By using the calculation in (\ref{Main deformed trace function, second expression}) we obtain
\begin{equation}\label{the pass to Upsilon}  
t^{-q}\varphi^L_{p,q}(tA)=\tr\bigl(t^{1-p}(I-H^*H+(p-1)L)+H^*A^{p-1}H\bigr)^{q/(p-1)}
\end{equation}
for $ t>0, $ Thus, by letting $ t\to 0 $ for $ p<1 $ or letting $ t\to \infty $ for $ p>1, $ we obtain that  the trace function 
\[
 A\to \tr\bigl(H^*A^{p-1}H\bigr)^{q/(p-1)}
 \]
is convex (respectively concave). It is no longer necessary to assume that $ H $ is a contraction, and since $ H^* A^{p-1}H $ is invertible, we can raise it to any non-zero exponent. 
Therefore, by inversion we obtain that the trace function
 \[
 A\to \tr\bigl(H^*A^{1-p}H\bigr)^{q/(1-p)}
 \]
 is convex (respectively concave) for arbitrary $ H. $ In particular, if $ H $ is a contraction we obtain by a small 
calculation the identity
\[
\varphi_{2-p,q}^{-L}(A)=\tr\bigl(I-H^*H+(p-1)L +H^*A^{1-p}H\bigr)^{q/(1-p)}.
\]
Since $ q/(1-p) >0 $ we obtain from Lemma~\ref{lemma: the symmetry} that also $ \varphi_{2-p,q}^{-L}(A) $ is convex (respectively concave).
\end{proof}

\subsection{The strategy of the proof}

We shall determine parameter values $ p $ and $ q $ such that $ F(X,A) $ is either concave/convex or just concave/convex in the second variable. To do this we use that the functions $ t\to t^p $ are operator concave, if and only if $ 0\le p\le 1, $ and operator convex, if and only if $ -1\le p\le 0 $ or $ 1\le p\le 2. $ It may be of interest to note that the same parameter conditions apply, if we only require matrix convexity or matrix concavity of order two, cf. \cite[Proposition 3.1]{hansen:2009:2}.
We also make use of Lieb's concavity theorem \cite[Corollary 1.1]{lieb:1973:1}  stating that the trace functions
\begin{equation}\label{Ando's and Lieb's trace functions}
(X,A)\to\tr X^p H^* A^q H
\end{equation}
are concave if $ p,q\ge 0 $ and $ p+q\le 1. $ Ando's theorem states \cite[Corollary 6.3]{ando:1979} that the trace function in (\ref{Ando's and Lieb's trace functions}) for an arbitrary matrix $ H, $ is convex for either $ -1\le p,q\le 0, $ or for
\[
 \quad -1\le p\le 0\quad\text{and}\quad 1-p\le q\le 2,
\]
where obviously $ p $ and $ q $ may be interchanged in the condition.
Since $ H=I $ is a possibility, we realise that concavity of $ \varphi_{p,q}(A) $ requires $ 0\le q\le 1, $ while convexity of $ \varphi_{p,q}(A) $ requires $ q\le 0 $ or $ q\ge 1. $
Since we intend to eventually use operator convexity/concavity of the function $ t\to t^p, $   we are restricted to the cases
\[
-1\le p-1\le 2\quad\text{or equivalently}\quad 0\le p\le 3.
\]
Note that if $ \beta=1, $ then $ p=1. $ 

 \begin{proposition}\label{proposition: concavity of operator function}

Let $ K $ be a positive definite $ n\times n $ matrix, and let $ H $ be any $ n\times n $ matrix. We may define the operator map
\[
\psi_p^s(A)=\bigl(K+H^* A^p H\bigr)^s
\]
in positive definite $ n\times n $ matrices for exponents $ p $ and $ s. $ If $ -1\le p\le 0 $ and $ -1\le s\le 0, $ then $ \psi_p^s(A) $ is concave.
\end{proposition}

\begin{proof}
We first consider the case $ -1\le p\le 0 $ and $ s=-1. $ Since
\[
\psi_p^{-1}(A)=\bigl(K+H^* A^p H\bigr)^{-1}=K^{-1/2}\bigl(I+L^*A^pL\bigr)^{-1}K^{-1/2},
\]
where $ L=HK^{-1/2}, $ we may assume $ K=I. $ We may also without loss generality assume that $ H $ is invertible. We then obtain
\[
\bigl(I+H^* A^p H\bigr)^{-1}=\frac{(H^*A^pH)^{-1}}{(H^*A^pH)^{-1}+I}=\frac{H^{-1}A^{-p}(H^{-1})^*}{H^{-1}A^{-p}(H^{-1})^*+I}
\]
by an elementary calculation. 
Since the map $ A\to H^{-1}A^{-p}(H^{-1})^* $ is concave and the function $ t\to t(1+t)^{-1} $ is operator monotone and operator concave, we obtain that  $ A\to \bigl(I+H^* A^p H\bigr)^{-1} $ is concave. That is, $ \psi_p^{-1}(A) $ is concave. Since the function $ t\to t^\alpha $ is both operator monotone and operator concave for $ 0\le\alpha\le 1, $ it follows that $ \psi_p^s(A) $ is concave for $ -1\le s\le 0. $
\end{proof}

\section{The main theorem}\label{Main theorem}

 \begin{theorem}\label{Main theorem}

The trace function $  \varphi^L_{p,q}(A) $ defined in (\ref{Main deformed trace function}) has the following geometric properties depending on the matrix $ L $ and the parameters $ p $ and $ q. $
\begin{subequations}
\begin{align}
\intertext{$ \varphi^L_{p,q}(A) $ is concave in positive definite $ A $  for} 
0\le &\,p\le 1, &&\,L\le 0, & 0\le &\,q\le 1.\label{Main theorem, first statement}\\
1\le &\,p\le 2,  &&\,L\ge 0, & 0\le &\,q\le 1.\label{Main theorem, second statement} \\
\intertext{$ \varphi^L_{p,q}(A) $ is convex in positive definite $ A $ for}
 0\le &\,p\le 1, &&\,L\le 0, & &\,q\le 0.\\
1\le &\,p\le 2,  &&\,L\ge 0, & &\,q\le 0.\\
2\le &\,p\le 3,   &&\,L\ge 0,  &&\,q\ge 1. 
 \end{align}
\end{subequations}

 \end{theorem}

 \begin{proof}
  We divide the proof following the statement's five cases.
  
\begin{enumerate}[(a)]

\item 
Take $ 0\le p< 1, $ $ L\le 0, $ and $ 1-p\le q\le 1. $ Then 
\[
0\le \beta=1+(p-1)/q \le p<1. 
\]
Since  $ \varphi_{p,q}(A)=\sup_{X>0} F(X,A), $   we may derive that $ \varphi_{p,q}(A) $ is concave if $ F(X,A) $ is jointly concave. Since $ 1<2-\beta\le 2 $ the first term in $ G(X,A) $ is concave (note that $ \beta<1). $ Since $ -1\le p-1\le 0 $ and $ 1-(p-1)\le 2-\beta\le 2, $ we realise by Ando's convexity theorem that $ \varphi_{p,q}(A) $ concave. 

Next take $ 0\le p\le 1, L\le 0, $ and $ 0\le q\le 1-p. $ That is $ -1\le s\le 0, $ where $ s=q/(p-1)\,. $ It then follows from  (\ref{Main deformed trace function, second expression}) combined with Corollary~\ref{proposition: concavity of operator function} that
$ \varphi_{p,q}(A) $ concave (even without the trace).

\item Take $ 1<p\le 2, $ $ L\ge 0, $ and $ 0< q\le p-1. $ Then 
\[
\beta=1+\frac{p-1}{q}\ge 2
\]
and thus $ \varphi^L_{p,q}(A)=\inf_{X>0} F(X,A). $ We may thus derive that $ \varphi^L_{p,q}(A) $ is concave if $ G(X,A) $ is concave in the second variable. This is so since $ 0\le p-1\le 1. $\\[1ex]
Next, take  $ 1<p\le 2, $ $ L\ge 0, $ and $ p-1\le q\le 1. $ Then 
\[
1< p\le \beta=1+\frac{p-1}{q}\le 2
\]
and thus $ \varphi^L_{p,q}(A)=\sup_{X>0} F(X,A). $ We may thus derive that $ \varphi^L_{p,q}(A) $ is concave if $ G(X,A) $ is concave. Since $ \beta>1 $ and $ 0\le 2-\beta\le 1, $ the first term in $ G(X,A) $ is concave. The second term is concave by Lieb's concavity theorem since $ 0\le p-1\le 1 $ and $ 2-\beta+p-1\le 1. $ The last inequality is satisfied since $ p\le\beta. $ These two cases taken together proves $ (b). $

\item
We first prove that $ (d) $ implies $ (c). $\\[1ex]
Take $ 1< p<2, $ $ L\ge 0, $ and $ q<0, $ and note that $ 0\le 2-p<1. $  Since $ \varphi^L_{p,q}(A) $ is convex by $ (d) $ and $ q/(1-p)>0, $  we obtain by Lemma~\ref{lemma: the symmetry} that $ \varphi^{-L}_{2-p,q}(A) $ is convex. This is equivalent to saying that  $ \varphi^L_{p,q}(A) $ is convex for $ 0\le p\le 1, $ $ L\le 0, $  and $ q< 0. $

\item 
By continuity we may assume $ 1< p\le 2, $ $ L\ge 0, $ and $ q<0. $ Therefore, $ \beta=1+(p-1)/q<1 $ and thus $ \varphi^L_{p,q}(A)=\sup_{X>0}F(X,A). $ We obtain that $ \varphi^L_{p,q}(A) $ is convex, if $ G(X,A) $ is convex in the second variable. This is so since $ \beta<1 $ and $ 0\le p-1\le 1. $

\item 
Take $ 2\le p\le 3, $ $ L\ge 0, $  and $ 1\le q\le p-1. $ Then 
\[
2\le\beta=1+(p-1)/q\le p\le 3. 
\]
Since $ \varphi^L_{p,q}(A)=\inf_{X>0} F(X,A), $ we may derive that $ \varphi^L_{p,q}(A) $ is convex if $ G(X,A) $ is convex. Since $ -1\le 2-\beta\le 0 $ this follows by Ando's convexity theorem if in addition
\[
1-(2-\beta)\le p-1\le 2,
\]
and this is satisfied since $ \beta\le p. $\\[1ex]
Take next $ 2\le p\le 3, $ $ L\ge 0, $ and $ q\ge p-1. $ Then
\[
1<\beta=1+(p-1)/q \le 2.
\]
Since $ \varphi^L_{p,q}(A)=\sup_{X>0} F(X,A), $ we may derive that $ \varphi^L_{p,q}(A) $ is convex if $ G(X,A) $ is convex in the second variable. But this is so since $ \beta>1 $ and $ 1\le p-1\le 2. $ These two cases taken together proves $ (e). $
 \end{enumerate}
  \end{proof}
  
  The special case $ q=1 $ was proved in \cite[Corollary 2.3]{shi-hansen:2020}.

 \subsection{Comparison with the literature}

The trace functions $ \Upsilon_{p,q}(A) $ were introduced and studied by Carlen and Lieb in \cite[Theorem 1.1]{carlen:2008} and later with a different definition (by setting $ s=q/p) $ in \cite[Proposition 5]{Carlen-Lieb:2018}. We adopt and slightly generalise the latter definition by setting
\begin{equation}\label{generalised Upsilon function}
\Upsilon^K_{p,s}(A)=\tr\bigl(K+H^*A^pH\bigr)^s,
\end{equation}
where $ K\ge 0, $ $ H $ is arbitrary, and $ A $ is positive definite. By replacing $ p $ with $ p-1 $ and applying Corollary~\ref{the connection with Upsilon},  we obtain the following corollary to Theorem~\ref{Main theorem}.

\begin{corollary}\label{corollary to main theorem}
The trace function $  \Upsilon^K_{p,s}(A) $ defined in (\ref{generalised Upsilon function}) has the following geometric properties depending on the matrix $ K $ and the parameters $ p $ and $ s. $
\begin{subequations}
\begin{align}
\intertext{$ \Upsilon^K_{p,q}(A) $ is concave in positive definite $ A $ for}
-1\le&\,p\le 0,& p^{-1}\le &\,s\le 0.\label{first concavity case}\\[1ex]
0\le &\,p\le 1,&  0\le &\,s\le p^{-1}.\label{second concavity case}\\[1ex]
\intertext{$ \Upsilon^K_{p,q}(A) $ is convex in positive definite $ A $ for}
-1\le&\,p\le 0,& &\,s\ge 0.\label{first convexity case}\\[1ex]
0\le&\,p\le 1,& &\,s\le 0.\label{second convexity case}\\[1ex]
1\le &\,p\le 2,& &\,s\ge p^{-1}. 
\end{align}
\end{subequations}

\end{corollary}

\begin{proof} To a given $ K\ge 0 $ we set $ L=(p-1)^{-1} K. $ Then $ L\le 0 $ for $ p<1 $ and $ L\ge 0 $ for $ p>1. $  By replacing $ L $ with $ t^{p-1} L $ in equation (\ref{the pass to Upsilon}) we obtain
\[
t^{-q}\varphi^{t^{p-1}L}_{p,q}(tA)=\tr\bigl( t^{1-p}(I-H^*H)+(p-1)L+H^*A^{p-1}H\bigr)^{q/(p-1)}
\]
and this expression tends to 
\[
\tr\bigl(K+H^*A^{p-1}H\bigr)^{q/(p-1)}= \Upsilon^K_{p-1,s}(A),\quad s=\frac{q}{p-1}
\]
by letting $ t\to 0 $ in the case $ p<1, $ and letting $ t\to\infty $ in the case $ p>1. $ With these choices we realise that $ \Upsilon^K_{p-1,s}(A) $ has the same geometric properties as $ \varphi^L_{p,q}(A). $ We may now replace $ p $ with $ p+1 $ and obtain that $ \Upsilon^K_{p,s}(A) $ has the same geometric properties as does $ \varphi^L_{p+1,q}(A), $ where $ s=q/p\,. $ 

In particular,
$ \Upsilon^K_{p,s}(A) $ is concave for $ -1\le p\le 0 $ and $ 0\le q\le 1, $ equivalent to $ p^{-1}\le s \le 0. $ 
Furthermore, $ \Upsilon^K_{p,s}(A) $ is convex for $ -1\le p\le 0 $ and $ q\le 0, $ equivalent to $ s\ge 0. $ Likewise, $ \Upsilon^K_{p,s}(A) $ is convex for $ 0\le p\le 1 $ and $ q\le 0, $ equivalent to $ s\le 0. $ Finally,  $ \Upsilon^K_{p,s}(A) $ is convex for $ 1\le p\le 2 $ and $ q\ge 1, $ equivalent to $ s\ge p^{-1}.   $
\end{proof}

In the case $ K=0 $  we note that (\ref{first concavity case}) and (\ref{second concavity case}) are counterparts of each other by replacing $ (p,s) $ with $ (-p,-s). $ This also applies to  (\ref{first convexity case}) and (\ref{second convexity case}).

These results contains the statements in \cite[Proposition 5]{Carlen-Lieb:2018}, where the authors list the following clarifications. 
\begin{enumerate}[(1)]

\item Concavity: The case $ 0\le p\le 1 $ and $ K=0. $  The result for $ s=p^{-1} $  is due to Epstein \cite{epstein:1973}. Carlen and Lieb proved the result for $ 1\le s\le p^{-1}, $ \cite[Theorem 1.1]{carlen:2008}. The full result for $ 0\le s\le p^{-1} $ is due to Hiai \cite[Theorem 4.1 (1)]{Hiai-LAA:2013}.

\item Convexity: The case $ -1\le p\le 0, $ $ K=0, $ and $ s>0 $ is due to Hiai \cite[Theorem 4.1 (2)]{Hiai-LAA:2013}.

\item Convexity: The case $ 1\le p\le 2, $ K=0,  and $ s\ge p^{-1} $ is due to Carlen and Lieb \cite[Theorem 1.1]{carlen:2008}.

\end{enumerate}  

The dual case $ 0\le p\le 1, $ K=0, and $ s<0 $ is also contained in  Hiai \cite[Theorem 4.1 (2)]{Hiai-LAA:2013}.
One may compare Corollary~\ref{corollary to main theorem} to Figure 1.1 in Zhang \cite{Zhang:2020}, where a three variable extension of $ \Psi_{p,q,s}(A) $ of $ \Upsilon^0_{p,s}(A) $ is discussed. The comparison is obtained by setting $ q=0  $ in the figure.\\[1ex]

\noindent\textbf{Declaration of interest}: The author has no conflict of interest.\\[1ex]
\noindent\textbf{No datasets were generated or analysed during the current study.}


\small{

\vfill

\noindent Frank Hansen: Department of Mathematical Sciences, Copenhagen University, Denmark.\\
Email: frank.hansen@math.ku.dk.
      }

\end{document}